\documentclass[final,1p,times,authoryear]{elsarticle}
\usepackage{amsmath}
\usepackage{amsthm}
\usepackage{amssymb}
\usepackage{amsfonts}
\usepackage{dsfont}

\usepackage{float}
\usepackage[plainpages=false,pdfpagelabels=true,colorlinks=true,citecolor=blue,hypertexnames=false]{hyperref}
\usepackage[ruled,vlined]{algorithm2e}
\usepackage{mathtools}
\usepackage{extarrows}
\usepackage{url}
\usepackage{balance}
\usepackage{multirow}
\usepackage{color}
\usepackage{diagbox}
\usepackage{booktabs}
\usepackage{makecell}
\newcolumntype{V}{!{\vrule width 1pt}}

\newtheorem{theorem}{Theorem}
\newtheorem{lemma}[theorem]{Lemma}

\newtheorem{proposition}[theorem]{Proposition}
\newtheorem{definition}[theorem]{Definition}

\newtheorem{example}[theorem]{Example}
\newtheorem{remark}[theorem]{Remark}
\newtheorem{problem}[theorem]{Problem}

\newcommand{\gr}{Gr\"{o}bner }
\def\z{{\bf z}}
\def\F{{\mathbf{F}}}
\def\G{{\mathbf{G}}}

\begin{document}

\begin{frontmatter}

\title{On Factor Left Prime Factorization Problems for Multivariate Polynomial Matrices}

\author[baic,smbh]{Dong Lu}
\ead{donglu@buaa.edu.cn}

\author[klmm,ucas]{Dingkang Wang}
\ead{dwang@mmrc.iss.ac.cn}

\author[klmm,ucas]{Fanghui Xiao\corref{cor1}}
\ead{xiaofanghui@amss.ac.cn}

\cortext[cor1]{Corresponding author}

\address[baic]{Beijing Advanced Innovation Center for Big Data and Brain Computing, Beihang University, Beijing 100191, China}

\address[smbh]{School of Mathematical Sciences, Beihang University, Beijing 100191, China}

\address[klmm]{KLMM, Academy of Mathematics and Systems Science, Chinese Academy of Sciences, Beijing 100190, China}

\address[ucas]{School of Mathematical Sciences, University of Chinese Academy of Sciences, Beijing 100049, China}

\begin{abstract}
 This paper is concerned with factor left prime factorization problems for multivariate polynomial matrices without full row rank. We propose a necessary and sufficient condition for the existence of factor left prime factorizations of a class of multivariate polynomial matrices, and then design an algorithm to compute all factor left prime factorizations if they exist. We implement the algorithm on the computer algebra system Maple, and two examples are given to illustrate the effectiveness of the algorithm. The results presented in this paper are also true for the existence of factor right prime factorizations of multivariate polynomial matrices without full column rank.
\end{abstract}

\end{frontmatter}

\section{Introduction}\label{intro}

 The factorization problems of multivariate polynomial matrices have attracted much attention over the past decades because of their fundamental importance in multidimensional systems, circuits, signal processing, controls, and other related areas \citep{Bose1982,Bose2003}. Up to now, the factorization problems have been solved for univariate and bivariate polynomial matrices \citep{Guiver1982Polynomial,Morf1977New}. However, there are still many challenging open problems for multivariate (more than two variables) polynomial matrix factorizations due to the lack of a mature polynomial matrix theory.

 \cite{Youla1979Notes} studied the basic structure of multidimensional systems theory, and proposed three types of factorizations for multivariate polynomial matrices: zero prime factorization, minor prime factorization and factor prime factorization. The existence problem of zero prime factorizations for multivariate polynomial matrices with full rank first raised in \citep{Lin1999Notes}, and has been solved in \citep{Pommaret2001Solving,Wang2004On}. In recent years, the factorization problems of multivariate polynomial matrices without full rank deserve some attention. \cite{Lin2001A} studied a generalization of Serre's conjecture, and they pointed out some relationships between the existence of a zero prime factorization for a multivariate polynomial matrix without full rank and its an arbitrary full rank submatrix.

 \cite{Mingsheng2005On} completely solved the existence problem of minor prime factorizations for multivariate polynomial matrices with full rank, and proposed an effective algorithm. \cite{Guan2019} extended the main result in \citep{Mingsheng2005On} to the case of non-full rank. In order to study the existence problem of factor prime factorizations for multivariate polynomial matrices with full rank, \cite{Mingsheng2007On} proposed the concept of regularity and obtained a necessary and sufficient condition. \cite{Guan2018} gave an algorithm to determine whether a class of multivariate polynomial matrices without full rank has factor prime factorizations.

 Although some achievements have been made on the existence for factor prime factorizations of some classes of multivariate polynomial matrices, factor prime factorizations are still open problems. Therefore, we focus on factor left prime factorization problems for multivariate polynomial matrices without full row rank in this paper.

 The rest of the paper is organized as follows. In section \ref{sec_PP}, we introduce some basic concepts and present the two major problems on factor left prime factorizations. We present in section \ref{sec_MR} a necessary and sufficient condition for the existence of factor left prime factorizations of a class of multivariate polynomial matrices without full row rank. In section \ref{sec_AE}, we construct an algorithm and use two examples to illustrate the effectiveness of the algorithm. We end with some concluding remarks in section \ref{sec_conclusions}.

\section{Preliminaries and Problems}\label{sec_PP}

 We denote by $k$ an algebraically closed field, $\z$ the $n$ variables $z_1,\ldots,z_n$ where $n\geq 3$. Let $k[\z]$ be the polynomial ring, and $k[\z]^{l\times m}$ be the set of $l\times m$ matrices with entries in $k[\z]$. Throughout this paper, we assume that $l\leq m$. In addition, we use ``w.r.t." to represent ``with respect to".

 For any given polynomial matrix $\mathbf{F}\in k[\z]^{l\times m}$, let ${\rm rank}(\F)$ and $\mathbf{F}^{\rm T}$ be the rank and the transposed matrix of $\mathbf{F}$, respectively; if $l = m$, we use  ${\rm det}(\mathbf{F})$ to denote the determinant of $\mathbf{F}$; we denote by $\rho(\F)$ the submodule of $k[\z]^{1\times m}$ generated by the rows of $\F$; for each $i$ with $1\leq i \leq {\rm rank}(\F)$, let $d_i(\F)$ be the greatest common divisor of all the $i\times i$ minors of $\mathbf{F}$; let ${\rm Syz}(\F)$ be the syzygy module of $\mathbf{F}$, i.e., ${\rm Syz}(\F) = \{ \vec{v}\in k[\z]^{m\times 1} : \F\vec{v} = \vec{0} \}$.


\subsection{Basic Notions}

 The following three concepts, which were first proposed in \citep{Youla1979Notes}, play an important role in multidimensional systems.

 \begin{definition}
  Let $\mathbf{F}\in k[\z]^{l\times m}$ be of full row rank.
  \begin{enumerate}
    \item If all the $l\times l$ minors of $\mathbf{F}$ generate $k[\z]$, then $\mathbf{F}$ is said to be a zero left prime (ZLP) matrix.

    \item If all the $l\times l$ minors of $\mathbf{F}$ are relatively prime, i.e., $d_l(\mathbf{F})$ is a nonzero constant, then $\mathbf{F}$ is said to be an minor left prime (MLP) matrix.

    \item If for any polynomial matrix factorization $\mathbf{F} = \mathbf{F}_1\mathbf{F}_2$ in which $\mathbf{F}_1\in k[\z]^{l\times l}$, $\mathbf{F}_1$ is necessarily a unimodular matrix, i.e., ${\rm det}(\mathbf{F}_1)$ is a nonzero constant, then $\mathbf{F}$ is said to be a factor left prime (FLP) matrix.
  \end{enumerate}
 \end{definition}

 Let $\mathbf{F}\in k[\z]^{m\times l}$ with $m\geq l$, then a ZRP (MRP, FRP) matrix can be similarly defined. Note that ZLP $\Rightarrow$ MLP $\Rightarrow$ FLP. Youla and Gnavi proved that when $n=1$, the three concepts coincide; when $n=2$, ZLP is not equivalent to MLP, but MLP is the same as FLP; when $n\geq 3$, these concepts are pairwise different.

 A factorization of a multivariate polynomial matrix is formulated as follows.

 \begin{definition}\label{matrix_factorization}
  Let $\mathbf{F}\in k[\z]^{l\times m}$ with rank $r$ and $f$ is a divisor of $d_r(\F)$, where $1\leq r \leq l$. $\mathbf{F}$ is said to admit a factorization w.r.t. $f$ if $\mathbf{F}$ can be factorized as
  \begin{equation}\label{gerneral-matirx-factorization}
   \mathbf{F} = \mathbf{G}_1\mathbf{F}_1
  \end{equation}
  such that $\mathbf{F}_1\in k[\z]^{r\times m}$, $\mathbf{G}_1\in k[\z]^{l\times r}$ with $d_r(\mathbf{G}_1) = f$. In particular, Equation (\ref{gerneral-matirx-factorization}) is said to be a ZLP (MLP, FLP) factorization of $\F$ w.r.t. $f$ if $\F_1$ is a ZLP (MLP, FLP) matrix.
 \end{definition}

 In order to state conveniently problems and main conclusions of this paper, we introduce the following concepts and results.

 \begin{definition}\label{quotient-define}
  Let $\mathcal{K}$ be a submodule of $k[\z]^{1\times m}$, and $J$ be an ideal of $k[\z]$. We define $\mathcal{K} : J= \{ \vec{u}\in k[\z]^{1\times m} : J \vec{u} \subseteq \mathcal{K} \}$, where $J \vec{u}$ is the set $\{f\vec{u} : f\in J\}$.
 \end{definition}

 Obviously, $\mathcal{K} \subseteq \mathcal{K} : J$. Let $I \subset k[\z]$ be another ideal, it is easy to show that
 \begin{equation}\label{quotient-module-2}
   \mathcal{K} : (IJ) = (\mathcal{K} : I) : J.
 \end{equation}
 Equation (\ref{quotient-module-2}) is a simple generalization of Proposition 10 in subsection $4$, Zariski closure and quotients of ideals in \citep{Cox2007Ideals}. For convention, we write $\mathcal{K} :\langle f \rangle$ as $\mathcal{K} : f$ for any $f\in k[\z]$.

 \begin{definition}\label{torsion-define}
  Let $\mathcal{K}$ be a $k[\z]$-module. The torsion submodule of $\mathcal{K}$ is defined as ${\rm Torsion}(\mathcal{K}) = \{\vec{u}\in \mathcal{K}  :  \exists f \in k[\z] \backslash \{0\} \text{ such that } f\vec{u} = \vec{0} \}$.
 \end{definition}

 We refer to \citep{Eisenbud2013} for more details about the above two concepts. Let $\mathcal{K}_1,\mathcal{K}_2$ be two $k[\z]$-modules, we define $\mathcal{K}_1/\mathcal{K}_2 = \{\vec{u}+ \mathcal{K}_2 : \vec{u} \in \mathcal{K}_1\}$. \cite{Liu2015Further} established a relationship between Definition \ref{quotient-define} and Definition \ref{torsion-define}.

 \begin{lemma}\label{LW-Torsion}
  Let $\F\in k[\z]^{l\times m}$ be of full row rank, $d = d_l(\F)$ and $\mathcal{K} = \rho(\F)$. Then $(\mathcal{K}:d)/\mathcal{K} = {\rm Torsion}(k[\z]^{1\times m}/\mathcal{K})$.
 \end{lemma}

  Moreover, Liu and Wang further extended the Youla's MLP lemma, which had been used to give another proof of the Serre's problem.

 \begin{lemma}\label{Serre-LW}
  Let $\F\in k[\z]^{l\times m}$ be of full row rank and $d = d_l(\F)$. Then for each $i=1,\ldots,n$, there exists $\mathbf{V}_i\in k[\z]^{m\times l}$ such that $\F\mathbf{V}_i = d\varphi_i \mathbf{I}_{l\times l}$, where $\varphi_i$ is nonzero and independent of $z_i$.
 \end{lemma}

 \cite{Guan2018} proved the following lemma, which is similar to the above result.

 \begin{lemma}\label{minor-Guan}
  Let $\G\in k[\z]^{l\times r}$ be of full column rank with $l \geq r$, and $g$ be an arbitrary $r\times r$ minor of $\G$. Then there exists $\G'\in k[\z]^{r\times l}$ such that $\G'\G = g\mathbf{I}_{r\times r}$.
 \end{lemma}

 In order to study the properties of multivariate polynomial matrices, \cite{Lin1988On} and \cite{Sule1994Feed} introduced the following important concept.

 \begin{definition}
  Let $\mathbf{F}\in k[\z]^{l\times m}$ with rank $r$, where $1\leq r \leq l$. For any given integer $i$ with $1\leq i \leq r$, let $a_1,\ldots,a_\beta$ denote all the $i\times i$ minors of $\mathbf{F}$, where $\beta = \binom l{i} \cdot \binom m{i}$. Extracting $d_i(\mathbf{F})$ from $a_1,\ldots,a_\beta$ yields
  \[a_j = d_i(\mathbf{F})\cdot b_j, ~ j=1,\ldots,\beta.\]
  Then, $b_1,\ldots,b_\beta$ are called all the $i\times i$ reduced minors of $\mathbf{F}$.
 \end{definition}

 \cite{Lin1988On} showed that reduced minors are important invariants for multivariate polynomial matrices.

 \begin{lemma}\label{RM_relation}
  Let $\F_1\in k[\z]^{r\times t}$ be of full row rank, $b_1, \ldots, b_{\gamma}$ be all the $r\times r$ reduced minors of $\F_1$, and $\F_2\in k[\z]^{t\times (t-r)}$ be of full column rank, $\bar{b}_1, \ldots, \bar{b}_{\gamma}$ be all the $(t-r)\times (t-r)$ reduced minors of $\F_2$, where $r<t$ and $\gamma = \binom {t}{r}$. If $\F_1\F_2 = \mathbf{0}_{r\times(t-r)}$, then $\bar{b}_i=\pm b_i$ for $i=1,\ldots,\gamma$, and signs depend on indices.
 \end{lemma}

 Let $\mathbf{F}\in k[\z]^{l\times m}$ with rank $r$, where $1\leq r < l$. Let $\bar{\F}_1,\ldots,\bar{\F}_\eta \in k[\z]^{l\times r}$ be all the full column rank submatrices of $\F$, where $1\leq \eta \leq \binom{m}{r}$. According to Lemma \ref{RM_relation}, it follows that $\bar{\F}_1,\ldots,\bar{\F}_\eta$ have the same $r\times r$ reduced minors. Based on this phenomenon, we give the following concept which was first proposed in \citep{Lin2001A}.

 \begin{definition}
  Let $\mathbf{F}\in k[\z]^{l\times m}$ with rank $r$, and $\bar{\F}\in k[\z]^{l\times r}$ be an arbitrary full column rank submatrix of $\F$, where $1\leq r < l$. Let $c_1,\ldots, c_\xi$ be all the $r\times r$ reduced minors of $\bar{\F}$, where $\xi = \binom{l}{r}$. Then $c_1,\ldots, c_\xi$ are called all the $r\times r$ {\bf column} reduced minors of $\F$.
 \end{definition}

 The above concept will play an important role in this paper. Obviously, the calculation amount of all the $r\times r$ column reduced minors of $\F$ is much less than that of all the $r\times r$ reduced minors of $\F$ in general.

 \begin{lemma}\label{QS-theorem}
  Let $\mathbf{U}\in k[\z]^{l \times m}$ be a ZLP matrix, where $l<m$. Then there exists a ZRP matrix $\mathbf{V}\in k[\z]^{m \times l}$ such that $\mathbf{U}\mathbf{V} = \mathbf{I}_{l\times l}$. Moreover, ${\rm Syz}(\mathbf{U})$ is a free submodule of $k[\z]^{m\times 1}$ with rank $m-l$.
 \end{lemma}

 The above result is called the Quillen-Suslin theorem. In order to solve the problem whether any finitely generated projective module over a polynomial ring is free, \cite{Quillen1976Projective} and  \cite{Suslin1976Projective} solved the problem positively and independently.

 Using the Quillen-Suslin theorem, \cite{Pommaret2001Solving} and \cite{Wang2004On} solved the Lin-Bose conjecture.

 \begin{lemma}\label{Lin-Bose-conjecture}
  Let $\F\in k[\z]^{l \times m}$ be of full row rank, where $l<m$. If all the $l\times l$ reduced minors of $\F$ generate $k[\z]$, then $\F$ has a ZLP factorization.
 \end{lemma}

 Let $\mathbf{F}\in k[\z]^{l\times m}$ be of full row rank, and $f$ be a divisor of $d_l(\mathbf{F})$. In order to study a factorization of $\mathbf{F}$ w.r.t. $f$, \cite{Mingsheng2007On} introduced the concept of regularity. $f$ is said to be regular w.r.t. $\F$ if and only if $d_l([f\mathbf{I}_{l\times l} ~ \F]) = f$ up to multiplication by a nonzero constant. Then, Wang obtained the following result.

 \begin{lemma}\label{W-flp}
  Let $\F\in k[\z]^{l\times m}$ be of full row rank, and $f$ be regular w.r.t. $\F$. Then $\F$ has a factorization w.r.t. $f$ if and only if $\rho(\F):f$ is a free module of rank $l$.
 \end{lemma}

\subsection{Problems}

 According to Lemma \ref{W-flp}, Wang proposed a necessary and sufficient condition to verify whether $\F$ has a FLP factorization w.r.t. $f$. After that, \cite{Guan2018} considered the case of multivariate polynomial matrices without full row rank. When $f$ satisfies a special property, they obtained a necessary condition that $\mathbf{F}$ has a factorization w.r.t. $f$, and designed an algorithm to compute all FLP factorizations of $\F$ if they exist. In this paper we will further consider the following two problems concerning FLP factorizations.

 \begin{problem}\label{main-problem-1}
  Let $\mathbf{F}\in k[\z]^{l\times m}$ with rank $r$, and $f$ be a divisor of $d_r(\mathbf{F})$, where $1\leq r < l$. Determine whether $\F$ has a FLP factorization w.r.t. $f$.
 \end{problem}

 \begin{problem}\label{main-problem-2}
  Let $\mathbf{F}\in k[\z]^{l\times m}$ with rank $r$, where $1\leq r < l$. Constructing an algorithm to compute all FLP factorizations of $\F$.
 \end{problem}

 Youla and Gnavi used an example to show that it is very difficult to judge whether a multivariate polynomial matrix is a FLP matrix. Hence, Problem \ref{main-problem-1} and Problem \ref{main-problem-2} may be very difficult in general. In this paper, we will give partial solutions to the above two problems.

\section{Main Results}\label{sec_MR}

 Let $\mathbf{F}\in k[\z]^{l\times m}$ with rank $r$, and $f$ be a divisor of $d_r(\F)$, where $1\leq r < l$. We use the following lemma to illustrate that all the $r\times r$ column reduced minors of $\F$ play an important role in a factorization of $\F$ w.r.t. $f$.

 \begin{lemma}\label{comple-lemma}
  Let $\mathbf{F}\in k[\z]^{l\times m}$ with rank $r$, $f$ be a divisor of $d_r(\F)$, and $c_1,\ldots, c_\xi$ be all the $r\times r$ column reduced minors of $\F$, where $1\leq r < l$. If there exist $\mathbf{G}_1\in k[\z]^{l\times r}$ and $\mathbf{F}_1\in k[\z]^{r\times m}$ such that $\mathbf{F} = \mathbf{G}_1\mathbf{F}_1$ with $d_r(\G_1)=f$, then $I_r(\G_1) = \langle fc_1,\ldots, fc_\xi \rangle$.
 \end{lemma}

 \begin{proof}
  Since $\mathbf{F}$ is a matrix with rank $r$, there exists a full row rank matrix $\mathbf{A}\in k[\z]^{(l-r)\times l}$ such that $\mathbf{A}\mathbf{F} = \mathbf{0}_{(l-r)\times m}$. Let $\bar{\F} \in k[\z]^{l\times r}$ be an arbitrary full column rank submatrix of $\F$, then $\mathbf{A}\bar{\mathbf{F}} = \mathbf{0}_{(l-r)\times r}$. Based on Lemma \ref{RM_relation}, all the $r\times r$ reduced minors of $\mathbf{A}$ are $c_1,\ldots, c_\xi$. It follows from ${\rm rank}(\F) \leq {\rm min}\{{\rm rank}(\G_1),{\rm rank}(\F_1)\}$ that $\G_1$ is a full column rank matrix and $\F_1$ is a full row rank matrix. Then $\mathbf{A}\mathbf{G}_1\mathbf{F}_1 =\mathbf{0}_{(l-r)\times m}$ implies that $\mathbf{A}\mathbf{G}_1 =\mathbf{0}_{(l-r)\times r}$. Using Lemma \ref{RM_relation} again, all the $r\times r$ reduced minors of $\G_1$ are $c_1,\ldots, c_\xi$. Consequently, $I_r(\G_1) = \langle fc_1,\ldots, fc_\xi \rangle$ since $d_r(\G_1)=f$. \qed
 \end{proof}

 Now, we give the first main result in this paper.

 \begin{theorem}\label{LWX-theorem-1}
  Let $\mathbf{F}\in k[\z]^{l\times m}$ with rank $r$, $f$ be a divisor of $d_r(\F)$ and $c_1,\ldots, c_\xi$ be all the $r\times r$ column reduced minors of $\F$, where $1\leq r < l$. Let $d = d_r(\F)$ and $\mathcal{K} = \rho(\F)$, then the following are equivalent:
  \begin{enumerate}
    \item $\F$ has a factorization w.r.t. $f$;

    \item there exists $\F_1\in k[\z]^{r\times m}$ with full row rank such that $d_r(\F_1)= \frac{d}{f}$ and $\mathcal{K} \subseteq \rho(\F_1) \subseteq \mathcal{K}:\langle fc_1,\ldots, fc_\xi \rangle$.
  \end{enumerate}
 \end{theorem}
 
 \begin{proof}
  $1\rightarrow 2$. Suppose that $\F$ has a factorization w.r.t. $f$. Then there exist $\G_1\in k[\z]^{l\times r}$ and $\F_1\in k[\z]^{r\times m}$ such that $\F = \G_1\F_1$ with $d_r(\G_1) = f$. Clearly, $\mathcal{K} \subseteq \rho(\F_1)$. From $d_r(\F) = d_r(\G_1)d_r(\F_1)$ we have $d_r(\F_1)= \frac{d}{f}$. According to Lemma \ref{comple-lemma}, $I_r(\G_1) = \langle fc_1,\ldots, fc_\xi \rangle$. Let $g$ be any $r\times r$ minor of $\G_1$, then there exists $\G'\in k[\z]^{r\times l}$ such that $\G'\G_1 = g\mathbf{I}_{r\times r}$ by Lemma \ref{minor-Guan}. Multiplying both left sides of $\mathbf{F} = \mathbf{G}_1\mathbf{F}_1$ by $\G'$, we get $\G'\mathbf{F} = \G'\mathbf{G}_1\mathbf{F}_1 = g \mathbf{F}_1$. This implies that $g\cdot \rho(\F_1) \subseteq \mathcal{K}$. Noting that $g$ is an arbitrary $r\times r$ minor of $\G_1$, we obtain $\rho(\F_1)  \subseteq \mathcal{K} : I_r(\G_1) = \mathcal{K} : \langle fc_1,\ldots, fc_\xi \rangle$.

  $2\rightarrow 1$. Thanks to $\mathcal{K} \subseteq \rho(\F_1)$, there exists $\G_1\in k[\z]^{l\times r}$ such that $\F = \G_1\F_1$. It follows from $d_r(\F) = d_r(\G_1)d_r(\F_1)$ that $d_r(\G_1)= f$. Then, $\F$ has a factorization w.r.t. $f$.  
 \end{proof}

 Although Theorem \ref{LWX-theorem-1} gives a necessary and sufficient condition for $\F$ to have a factorization w.r.t. $f$, it is difficult to find a full row rank matrix $\F_1\in k[\z]^{r\times m}$ that satisfies $d_r(\F_1)= \frac{d}{f}$ and $\mathcal{K} \subseteq \rho(\F_1) \subseteq \mathcal{K}:\langle fc_1,\ldots, fc_\xi \rangle$. Next, we will further study the relationship between $\rho(\F)$ and $\rho(\F_1)$.

 \begin{theorem}\label{LWX-module}
  Let $\mathbf{F}\in k[\z]^{l\times m}$ with rank $r$, $f$ be a divisor of $d_r(\F)$ and $c_1,\ldots, c_\xi$ be all the $r\times r$ column reduced minors of $\F$, where $1\leq r < l$. Suppose there exist $\mathbf{G}_1\in k[\z]^{l\times r}$ and $\mathbf{F}_1\in k[\z]^{r\times m}$ such that $\mathbf{F} = \mathbf{G}_1\mathbf{F}_1$ with $d_r(\G_1)=f$. Let $d = d_r(\F)$, $\mathcal{K} = \rho(\F)$ and $\mathcal{K}_1 = \rho(\F_1)$, then the following are equivalent:
  \begin{enumerate}
    \item $(\mathcal{K}_1:\frac{d}{f})/\mathcal{K}_1$;

    \item $(\mathcal{K}:\langle dc_1,\ldots, dc_\xi \rangle)/\mathcal{K}_1$;

    \item ${\rm Torsion}(k[\z]^{1\times m}/\mathcal{K}_1)$.
  \end{enumerate}
 \end{theorem}

 \begin{proof}
  It follows from ${\rm rank}(\F) \leq {\rm min}\{{\rm rank}(\G_1),{\rm rank}(\F_1)\}$ that $\F_1$ is a full row rank matrix. Since $d_r(\F) = d_r(\G_1)d_r(\F_1)$, we have $d_r(\F_1) = \frac{d}{f}$. It is apparent from Lemma \ref{LW-Torsion} that
  \begin{equation}\label{LWX-module-equ-1}
   (\mathcal{K}_1:\frac{d}{f})/\mathcal{K}_1 = {\rm Torsion}(k[\z]^{1\times m}/\mathcal{K}_1).
  \end{equation}

  If the following equation
  \begin{equation}\label{LWX-module-equ-0}
   \mathcal{K}_1:\frac{d}{f} =\mathcal{K}:\langle dc_1,\ldots, dc_\xi \rangle
  \end{equation}
  holds, then $(\mathcal{K}_1:\frac{d}{f})/\mathcal{K}_1$ and $(\mathcal{K}:\langle dc_1,\ldots, dc_\xi \rangle)/\mathcal{K}_1$ are obviously equivalent.

  We first verify $\mathcal{K}_1:\frac{d}{f} \subseteq \mathcal{K}:\langle dc_1,\ldots, dc_\xi \rangle$. Proceeding as in the proof of $1\rightarrow 2$ in Theorem \ref{LWX-theorem-1}, we get
  \begin{equation}\label{LWX-module-equ-8}
   \mathcal{K}_1  \subseteq \mathcal{K} : \langle fc_1,\ldots, fc_\xi \rangle.
  \end{equation}
  Using Equation (\ref{quotient-module-2}), we can derive
  \begin{equation}\label{LWX-module-equ-2}
   \mathcal{K}_1:\frac{d}{f} \subseteq (\mathcal{K}:\langle fc_1,\ldots, fc_\xi \rangle):\frac{d}{f} = \mathcal{K}:\langle dc_1,\ldots, dc_\xi \rangle.
  \end{equation}
  Next we show $\mathcal{K}:\langle dc_1,\ldots, dc_\xi \rangle \subseteq \mathcal{K}_1:\frac{d}{f}$. For any vector $\vec{u}\in \mathcal{K}:\langle dc_1,\ldots, dc_\xi \rangle = \bigcap_{j=1}^{\xi}(\mathcal{K}:dc_j)$, there exists $\vec{v}_j\in k[\z]^{1\times l}$ such that
  \begin{equation}\label{LWX-module-equ-3}
   dc_j\vec{u} = \vec{v}_j\F = \vec{v}_j\G_1\F_1, ~j =1,\ldots,\xi.
  \end{equation}
  Using Lemma \ref{Serre-LW}, for each $i=1,\ldots,n$, there exists $\mathbf{V}_i\in k[\z]^{m\times r}$ such that
  \begin{equation}\label{LWX-module-equ-4}
   \F_1\mathbf{V}_i = \frac{d}{f}\varphi_i \mathbf{I}_{r\times r},
  \end{equation}
  where $\varphi_i$ is nonzero and independent of $z_i$. Combining Equation (\ref{LWX-module-equ-3}) and Equation (\ref{LWX-module-equ-4}), we see that
  \begin{equation}\label{LWX-module-equ-5}
      dc_j \vec{u} \mathbf{V}_i
     = \vec{v}_j\mathbf{G}_1\mathbf{F}_1\mathbf{V}_i
     = \vec{v}_j\mathbf{G}_1(\frac{d}{f}\varphi_i \mathbf{I}_{r\times r})
     = \frac{d}{f}\varphi_i\vec{v}_j\mathbf{G}_1.
  \end{equation}
  As ${\rm gcd}(\varphi_1,\ldots,\varphi_n) = 1$, we have $dc_j\mid \frac{d}{f}\vec{v}_j\mathbf{G}_1$. This implies that $\frac{\vec{v}_j\mathbf{G}_1}{fc_j}$ is a polynomial vector. Then, it follows from Equation (\ref{LWX-module-equ-3}) that
  \begin{equation}\label{LWX-module-equ-6}
   \frac{d}{f}\vec{u} = \frac{\vec{v}_j\mathbf{G}_1}{fc_j}\mathbf{F}_1, ~j =1,\ldots,\xi.
  \end{equation}
  Thus, $\vec{u} \in \mathcal{K}_1 : \frac{d}{f}$, and we infer that  $\mathcal{K}:\langle dc_1,\ldots, dc_\xi \rangle \subseteq \mathcal{K}_1:\frac{d}{f}$.

  Consequently, $(\mathcal{K}_1:\frac{d}{f})/\mathcal{K}_1 =(\mathcal{K}:\langle dc_1,\ldots, dc_\xi \rangle)/\mathcal{K}_1$.  
 \end{proof}

 In Theorem \ref{LWX-module}, we obtain $\mathcal{K}_1:\frac{d}{f} = \mathcal{K}:\langle dc_1,\ldots, dc_\xi \rangle$. Naturally, we consider under what conditions $\mathcal{K}_1$ and $\mathcal{K}:\langle fc_1,\ldots, fc_\xi \rangle$ are equal. Now, we propose the following conclusion.

 \begin{theorem}\label{LWX-Guan}
  Let $\mathbf{F}\in k[\z]^{l\times m}$ with rank $r$, $f$ be a divisor of $d_r(\F)$ and $c_1,\ldots, c_\xi$ be all the $r\times r$ column reduced minors of $\F$, where $1\leq r < l$. Suppose there exist $\mathbf{G}_1\in k[\z]^{l\times r}$ and $\mathbf{F}_1\in k[\z]^{r\times m}$ such that $\mathbf{F} = \mathbf{G}_1\mathbf{F}_1$ with $d_r(\G_1) = f$. Let $d = d_r(\F)$, $\mathcal{K} = \rho(\F)$ and $\mathcal{K}_1 = \rho(\F_1)$. If ${\rm gcd}(f,\frac{d}{f}) = 1$, then $\mathcal{K}_1 = \mathcal{K} :\langle fc_1,\ldots,fc_\xi\rangle$ and $\mathcal{K} :\langle fc_1,\ldots,fc_\xi\rangle$ is a free module of rank $r$.
 \end{theorem}

 The above theorem is a generalization of Theorem 3.11 in \citep{Guan2018}. The proof of Theorem \ref{LWX-Guan} is basically the same as that of Theorem 3.11, except that we explicitly give a system of generators of $I_r(\G_1)$. Hence, the proof is omitted here. Evidently, the calculation amount of $\rho(\F_1) = \rho(\F) :\langle fb_1,\ldots,fb_\beta \rangle$ in Theorem 3.11 is much larger than that of $\rho(\F_1) = \rho(\F) :\langle fc_1,\ldots,fc_\xi\rangle$ in Theorem \ref{LWX-Guan}.

 Suppose ${\rm gcd}(f,\frac{d}{f}) = 1$. Let $\mathcal{K} :\langle fc_1,\ldots,fc_\xi\rangle$ be a free module of rank $r$, and a free basis of the module constitutes $\mathbf{F}_1\in k[\z]^{r\times m}$. Then, $\rho(\F_1) = \mathcal{K} :\langle fc_1,\ldots,fc_\xi\rangle$. Given $\mathcal{K} \subseteq \rho(\F_1)$, there exists $\mathbf{G}_1\in k[\z]^{l\times r}$ such that $\mathbf{F} = \mathbf{G}_1\mathbf{F}_1$ with $d_r(\G_1) = f'$, where $f'$ is a divisor of $d$. Notice that $f$ and $f'$ may be different. The condition that $\mathcal{K} :\langle fc_1,\ldots,fc_\xi\rangle$ is a free module of rank $r$ is only a necessary condition for the existence of a factorization of $\F$ w.r.t. $f$. In order to study the relationship between $f'$ and $f$, we first introduce a result in \citep{Liu2015Further}.

 \begin{lemma}\label{LW-constant}
  Let $\F\in k[\z]^{l\times m}$ be of full row rank, $d = d_l(\F)$ and $\mathcal{K} = \rho(\F)$. If there exists a divisor $f$ of $d$ such that $\mathcal{K} : f = \mathcal{K}$, then $f$ is a constant.
 \end{lemma}

 Now, we can draw the following conclusion.

 \begin{proposition}\label{LWX-Guan-equivalent}
  Let $\mathbf{F}\in k[\z]^{l\times m}$ with rank $r$, and $c_1,\ldots, c_\xi$ be all the $r\times r$ column reduced minors of $\F$, where $1\leq r < l$. Let $\mathcal{K} = \rho(\F)$, $d = d_r(\F)$ be a square-free polynomial and $f$ be a divisor of $d$. Suppose $\mathcal{K}_1 = \mathcal{K}:\langle fc_1,\ldots, fc_\xi \rangle$ is a free module of rank $r$ and $\mathbf{F}_1\in k[\z]^{r\times m}$ is composed of a free basis of $\mathcal{K}_1$. Then, there is no a proper divisor $f'$ of $f$ such that $\mathbf{F} = \mathbf{G}_1\mathbf{F}_1$, where $\mathbf{G}_1\in k[\z]^{l\times r}$ with $d_r(\G_1) = f'$.
 \end{proposition}

 \begin{proof}
  Note that $\mathcal{K} \subseteq \mathcal{K}_1$, there exists $\mathbf{G}_1\in k[\z]^{l\times r}$ such that $\mathbf{F} = \mathbf{G}_1\mathbf{F}_1$ with $d_r(\G_1) = f'$, where $f'$ is a divisor of $d$. Since $d$ is a square-free polynomial, ${\rm gcd}(f',\frac{d}{f'}) = 1$. According to Theorem \ref{LWX-Guan}, it follows that $\mathcal{K}_1 = \mathcal{K}:\langle f'c_1,\ldots, f'c_\xi \rangle$, i.e.,
  \begin{equation}\label{LWX-Guan-equivalent-equ-1}
   \mathcal{K}:\langle fc_1,\ldots, fc_\xi \rangle = \mathcal{K}:\langle f'c_1,\ldots, f'c_\xi \rangle.
  \end{equation}
  Assume that $f'$ is a proper divisor of $f$. It can easily be seen from Equation (\ref{LWX-Guan-equivalent-equ-1}) that
  \begin{equation}\label{LWX-Guan-equivalent-equ-2}
   \mathcal{K}_1 : \frac{f}{f'} = \mathcal{K}_1.
  \end{equation}
  Because $d_r(\F_1) = \frac{d}{f'}$, we have $\frac{f}{f'} \mid d_r(\F_1)$. Based on Lemma \ref{LW-constant}, $\frac{f}{f'}$ is a constant. This contradicts the fact that $f'$ is a proper divisor of $f$. 
 \end{proof}

 Before giving a new necessary and sufficient condition for the existence of a factorization of $\F$ w.r.t. $f$, we present the following result.

 \begin{lemma}\label{equivalent-zlp}
  Let $\mathbf{F}\in k[\z]^{l\times m}$ with rank $r$, and $c_1,\ldots, c_\xi$ be all the $r\times r$ column reduced minors of $\F$, where $1\leq r < l$. Then the following are equivalent:
  \begin{enumerate}
    \item there exist $\mathbf{U}\in k[\z]^{l\times r}$ and $\mathbf{F}_1\in k[\z]^{r\times m}$ such that $\mathbf{F} = \mathbf{U}\mathbf{F}_1$ with $\mathbf{U}$ being a ZRP matrix;

    \item $\langle c_1,\ldots, c_\xi \rangle =k[\z]$.
  \end{enumerate}
 \end{lemma}

 \begin{proof}
  $1\rightarrow 2$. Suppose there exist $\mathbf{U}\in k[\z]^{l\times r}$ and $\mathbf{F}_1\in k[\z]^{r\times m}$ such that $\mathbf{F} = \mathbf{U}\mathbf{F}_1$, where $\mathbf{U}$ is a ZRP matrix. Using Lemma \ref{comple-lemma}, $c_1,\ldots, c_\xi$ are all the $r\times r$ reduced minors of $\mathbf{U}$. Then, $\langle c_1,\ldots, c_\xi \rangle =k[\z]$ since $\mathbf{U}$ is a ZRP matrix.

  $2\rightarrow 1$. Because ${\rm rank}(\F)=r$, there exists a full row rank matrix $\mathbf{H}\in k[\z]^{(l-r)\times l}$ such that
  \begin{equation}\label{equivalent-zlp-equ-1}
   \mathbf{H}\F = \mathbf{0}_{(l-r)\times m}.
  \end{equation}
  According to Lemma \ref{RM_relation}, $c_1,\ldots, c_\xi$ are all the $(l-r)\times (l-r)$ reduced minors of $\mathbf{H}$. Assume that $\langle c_1,\ldots, c_\xi \rangle =k[\z]$. By Lemma \ref{Lin-Bose-conjecture}, $\mathbf{H}$ has a ZLP factorization
  \begin{equation}\label{equivalent-zlp-equ-2}
   \mathbf{H} = \mathbf{G}\mathbf{H}_1,
  \end{equation}
  where $\mathbf{G}\in k[\z]^{(l-r)\times (l-r)}$, and $\mathbf{H}_1\in k[\z]^{(l-r)\times l}$ is a ZLP matrix. Let $\vec{v}\in {\rm Syz}(\mathbf{H})$, then $\mathbf{H}\vec{v} = \mathbf{G}\mathbf{H}_1\vec{v} = \vec{0}$. Since $\mathbf{G}$ is a full column rank matrix, $\mathbf{H}_1\vec{v} = \vec{0}$. This implies that $\vec{v}\in {\rm Syz}(\mathbf{H}_1)$. Let $\vec{u}\in {\rm Syz}(\mathbf{H}_1)$, it is obvious that $\vec{u}\in {\rm Syz}(\mathbf{H})$. It follows that
  \begin{equation}\label{equivalent-zlp-equ-3}
   {\rm Syz}(\mathbf{H}) = {\rm Syz}(\mathbf{H}_1).
  \end{equation}
  Thus we conclude that ${\rm Syz}(\mathbf{H})$ is a free module of rank $r$ by the Quillen-Suslin theorem.

  Suppose that $\mathbf{U}\in k[\z]^{l\times r}$ is composed of a free basis of ${\rm Syz}(\mathbf{H})$. It follows from $\mathbf{H}\mathbf{U} = \mathbf{0}_{(l-r)\times r}$ that all the $r\times r$ reduced minors of $\mathbf{U}$ generate $k[\z]$. Using Lemma \ref{Lin-Bose-conjecture} again, there exist $\mathbf{U}_1\in k[\z]^{l\times r}$ and $\mathbf{G}_1\in k[\z]^{r\times r}$ such that
  \begin{equation}\label{equivalent-zlp-equ-4}
   \mathbf{U} = \mathbf{U}_1\mathbf{G}_1
  \end{equation}
  with $\mathbf{U}_1$ being a ZRP matrix. Since $\G_1$ is a full row rank matrix, from $\mathbf{H} \mathbf{U}_1\mathbf{G}_1 = \mathbf{0}_{(l-r)\times r}$ we have
  \begin{equation}\label{equivalent-zlp-equ-5}
   \mathbf{H} \mathbf{U}_1= \mathbf{0}_{(l-r)\times r}.
  \end{equation}
  This implies that
  \begin{equation}\label{equivalent-zlp-equ-6}
   \rho(\mathbf{U}_1^{\rm T}) \subseteq \rho(\mathbf{U}^{\rm T}).
  \end{equation}
  Using $d_r(\mathbf{U}) = d_r(\mathbf{U}_1){\rm det}(\G_1)$, we get $d_r(\mathbf{U}) = \delta{\rm det}(\G_1)$, where $\delta$ is a nonzero constant. If ${\rm det}(\G_1) \in k[\z] \setminus k$, then Equation (\ref{equivalent-zlp-equ-4}) implies that
  \begin{equation}\label{equivalent-zlp-equ-7}
   \rho(\mathbf{U}^{\rm T}) \subsetneq \rho(\mathbf{U}_1^{\rm T}).
  \end{equation}
  This leads to a contradiction. Thus, ${\rm det}(\G_1)$ is a nonzero constant. Consequently, we infer that $\mathbf{U}$ is a ZRP matrix.

  Equation (\ref{equivalent-zlp-equ-1}) implies that the columns of $\F$ belong to ${\rm Syz}(\mathbf{H})$, then there exists $\mathbf{F}_1\in k[\z]^{r\times m}$ such that
  \begin{equation}\label{equivalent-zlp-equ-8}
   \mathbf{F} = \mathbf{U}\mathbf{F}_1.
  \end{equation}
 \end{proof}

 Now, we give the second main result in this paper.

\vspace{4pt}
 \begin{theorem}\label{LWX-main-flp}
  Let $\mathbf{F}\in k[\z]^{l\times m}$ with rank $r$, and $c_1,\ldots, c_\xi$ be all the $r\times r$ column reduced minors of $\F$, where $1\leq r < l$. Let $\mathcal{K} = \rho(\F)$, $d = d_r(\F)$ and $f$ be a divisor of $d$ with ${\rm gcd}(f,\frac{d}{f}) = 1$. If $\langle c_1,\ldots, c_\xi \rangle =k[\z]$, then the following are equivalent:
  \begin{enumerate}
    \item $\F$ has a factorization w.r.t. $f$;

    \item $\mathcal{K} : f$ is a free module of rank $r$.
  \end{enumerate}
 \end{theorem}

 \begin{proof}
  $1\rightarrow 2$. Suppose that $\F$ has a factorization w.r.t. $f$. Then there exist $\mathbf{G}_1\in k[\z]^{l\times r}$ and $\mathbf{F}_1\in k[\z]^{r\times m}$ such that $\mathbf{F} = \mathbf{G}_1\mathbf{F}_1$ with $d_r(\G_1) = f$. According to Theorem \ref{LWX-Guan}, $\rho(\F_1) = \mathcal{K} : \langle fc_1,\ldots, fc_\xi \rangle$. It follows from $\langle c_1,\ldots, c_\xi \rangle =k[\z]$ that $\langle fc_1,\ldots, fc_\xi \rangle = \langle f \rangle$. Then, $\rho(\F_1) = \mathcal{K} : f$. As $\mathbf{F}_1$ is a full row rank matrix, $\mathcal{K} : f$ is a free module of rank $r$.

  $2\rightarrow 1$. Since $\langle c_1,\ldots, c_\xi \rangle =k[\z]$, by Lemma \ref{equivalent-zlp} we obtain
  \begin{equation}\label{LWX-main-flp-equ-1}
   \mathbf{F} = \mathbf{U}\mathbf{F}',
  \end{equation}
  where $\mathbf{U}\in k[\z]^{l\times r}$ is a ZRP matrix and $\mathbf{F}'\in k[\z]^{r\times m}$. Without loss of generality, we assume that $d_r(\mathbf{U}) =1$. Clearly, $\rho(\F) \subseteq \rho(\F')$. Based on the Quillen-Suslin theorem, there is a ZLP matrix $\mathbf{V} \in k[\z]^{r\times l}$ such that $\mathbf{V}\mathbf{U} = \mathbf{I}_{r\times r}$. Then, $\mathbf{F}' = \mathbf{V}\mathbf{F}$. This implies that $\rho(\F') \subseteq \rho(\F)$. Thus, $\rho(\F') = \mathcal{K}$, $d_r(\F') = d_r(\F)$ and $\rho(\F') : f$ is a free module of rank $r$. Since ${\rm gcd}(f,\frac{d}{f}) = 1$, $f$ is regular w.r.t. $\F'$. By Lemma \ref{W-flp}, there exist $\mathbf{G}'\in k[\z]^{r\times r}$ and $\mathbf{F}_1\in k[\z]^{r\times m}$ such that
  \begin{equation}\label{LWX-main-flp-equ-4}
   \F' = \G'\F_1
  \end{equation}
  with ${\rm det}(\G') = f$. By substituting Equation (\ref{LWX-main-flp-equ-4}) into Equation (\ref{LWX-main-flp-equ-1}), we get
  \begin{equation}\label{LWX-main-flp-equ-5}
   \F = (\mathbf{U}\G')\F_1.
  \end{equation}
  Let $\G_1 = \mathbf{U}\G'$, then $d_r(\G_1) = d_r(\mathbf{U}) {\rm det}(\G') = f$. Thus $\F$ has a factorization w.r.t. $f$. 
 \end{proof}

 \begin{remark}
 \cite{Mingsheng2007On} proved that $f$ is regular w.r.t. $\F'$ if ${\rm gcd}(f,\frac{d}{f}) = 1$.
 \end{remark}

 Let $\mathbf{F}\in k[\z]^{l\times m}$ with rank $r$ and $f$ be a divisor of $d_r(\F)$, where $1\leq r < l$. We define the following set:
 \[ M(f)=\{ h \in k[\z]: f \mid h \text{ and } h \mid d_r(\F) \}. \]

\vspace{4pt}

 Now, we give a partial solution to Problem \ref{main-problem-1}.

 \begin{theorem}\label{LWX-main-flp-2}
  Let $\mathbf{F}\in k[\z]^{l\times m}$ with rank $r$, and $c_1,\ldots, c_\xi$ be all the $r\times r$ column reduced minors of $\F$, where $1\leq r < l$. Let $\mathcal{K} = \rho(\F)$, $d = d_r(\F)$ and $f$ be a divisor of $d$. Suppose every $h\in M(f)$ satisfies ${\rm gcd}(h,\frac{d}{h}) = 1$ and $\langle c_1,\ldots, c_\xi \rangle =k[\z]$, then the following are equivalent:
  \begin{enumerate}
    \item $\F$ has a FLP factorization w.r.t. $f$;

    \item $\mathcal{K} : f$ is a free module of rank $r$, but $\mathcal{K} : h$ is not a free module of rank $r$ for every $h\in M(f)\setminus \{f \}$.
  \end{enumerate}
 \end{theorem}

 \begin{remark}
  With the help of Theorem \ref{LWX-main-flp}, the proof of Theorem  \ref{LWX-main-flp-2} is similar to that of Theorem 3.2 in \citep{Mingsheng2007On}, and is omitted here.
 \end{remark}

 In the above theorem, we need to verify whether a submodule of $k[\z]^{1\times m}$ is a free module of rank $r$. The traditional method is to calculate the $r$-th Fitting ideal of the submodule. We refer to \citep{Cox2005Using,Eisenbud2013,Greuel2002A} for more details. Next, we will give a simpler verification method.

\vspace{8pt}

 \begin{proposition}\label{free-module-check}
  Let $\mathbf{F}\in k[\z]^{l\times m}$ with rank $r$, and $J \subset k[\z]$ be a nonzero ideal, where $1\leq r < l$. Suppose $\mathbf{F}_0\in k[\z]^{s\times m}$ is composed of a system of generators of $\rho(\F): J$, then the following are equivalent:
  \begin{enumerate}
    \item $\rho(\F): J$ is a free module of rank $r$;

    \item all the $r\times r$ column reduced minors of $\F_0$ generate $k[\z]$.
  \end{enumerate}
 \end{proposition}

 \begin{proof}
   It is evident that $\rho(\F): J = \rho(\F_0)$. According to Proposition 3.14 in \citep{Guan2018}, the rank of $\rho(\F): J$ is $r$. This implies that ${\rm rank}(\F_0) = r$ and $s\geq r$.

  $1\rightarrow 2$. Suppose that $\rho(\F): J$ is a free module of rank $r$. Let $\mathbf{F}_1\in k[\z]^{r\times m}$ be composed of a free basis of $\rho(\F): J$, then $\rho(\F_1) = \rho(\F_0)$. On the one hand, $\rho(\F_0) \subseteq \rho(\F_1)$ implies that there exists $\mathbf{G}_1\in k[\z]^{s\times r}$ such that $\F_0 = \G_1\F_1$. On the other hand, it follows from $\rho(\F_1) \subseteq \rho(\F_0)$ that there exists $\mathbf{G}_0\in k[\z]^{r\times s}$ such that $\F_1 = \G_0\F_0$. Combining the above two equations, we have $\F_1 = (\G_0\G_1)\F_1$. Because $\F_1$ is a full row rank matrix, we obtain $\mathbf{I}_{r\times r} = \G_0\G_1$. According to the Binet-Cauchy formula, all the $r\times r$ minors of $\G_1$ generate $k[\z]$. Therefore, $\G_1$ is a ZRP matrix. Based on Lemma \ref{equivalent-zlp}, all the $r\times r$ column reduced minors of $\F_0$ generate $k[\z]$.

  $2\rightarrow 1$. There are two cases. First, $s>r$. Using Lemma \ref{equivalent-zlp}, there exist $\mathbf{F}_1\in k[\z]^{r\times m}$ and a ZRP matrix $\mathbf{U}\in k[\z]^{s\times r}$ such that $\mathbf{F}_0 = \mathbf{U}\mathbf{F}_1$. It follows from the proof of $2\rightarrow 1$ in Theorem \ref{LWX-main-flp} that $\rho(\F_0) = \rho(\F_1)$. Since $\F_1$ is a full row rank matrix, $\rho(\F): J$ is a free module of rank $r$. Second, $s=r$. In this situation, $\F_0$ is a full row rank matrix. This implies that $\rho(\F):J$ is a free module of rank $r$. Obviously, all the $r\times r$ column reduced minors of $\F_0$ are only one polynomial which is the constant $1$, and generate $k[\z]$. In summary, $\rho(\F):J$ is a free module of rank $r$.  
 \end{proof}

\section{Algorithm and Examples}\label{sec_AE}

\subsection{Algorithm}

 Before solving Problem \ref{main-problem-2}, we make the following analysis on the main results obtained in section \ref{sec_MR}. We first construct a polynomial matrix set of $k[\z]^{l\times m}$ as follows:
 \[\mathcal{M}=\{\F\in k[\z]^{l\times m} : d_r(\F) \text{ is a square-free polynomial}\},\]
 where $r = {\rm rank}(\F)$. Let $\mathbf{F}\in \mathcal{M}$, $d = d_r(\F)$, $\mathcal{K} = \rho(\F)$, $f$ be an arbitrary divisor of $d$, and $c_1,\ldots, c_\xi$ be all the $r\times r$ column reduced minors of $\F$, where $1\leq r < l$. There are two cases as follows.

 First, $\langle c_1,\ldots, c_\xi \rangle = k[\z]$. According to Theorem \ref{LWX-main-flp}, $\F$ has a factorization w.r.t. $f$ if and only if $\mathcal{K}:f$ is a free module of rank $r$. Since $f$ is an arbitrary divisor of $d$, we can compute all matrix factorizations of $\F$. After that, we obtain all FLP factorizations of $\F$ by Theorem \ref{LWX-main-flp-2}.

 Second, $\langle c_1,\ldots, c_\xi \rangle \neq k[\z]$. We only get a necessary condition for the existence of a factorization of $\F$ w.r.t. $f$ in Theorem \ref{LWX-Guan}. Nevertheless, we can get all factorizations of $\mathbf{F}$. The specific process is as follows. Let $f_1,\ldots,f_s$ be all different divisors of $d$ and $\mathcal{K}_j = \mathcal{K}: \langle f_jc_1,\ldots, f_jc_\xi \rangle$, then we verify whether $\mathcal{K}_j$ is a free module of rank $r$, where $j=1,\ldots,s$. For each $j$, one of the following three cases holds:
 \begin{enumerate}
   \item $\mathcal{K}_j$ is not a free module of rank $r$, then $\F$ has no factorization w.r.t. $f_j$;

   \item $\mathcal{K}_j$ is a free module of rank $r$, and a free basis of $\mathcal{K}_j$ constitutes $\mathbf{F}_j\in k[\z]^{r\times m}$,
       \begin{itemize}
        \item[2.1] if $d_r(\F_j) = \frac{d}{f_j}$, then $\F$ has a factorization w.r.t. $f_j$;

        \item[2.2] if $d_r(\F_j) \neq \frac{d}{f_j}$, then $\F$ has a factorization w.r.t. $f_i$, where $f_i\nmid f_j$.
       \end{itemize}
 \end{enumerate}
 Let $\F = \G_{i_1}\F_{i_1} =\cdots = \G_{i_t}\F_{i_t}$ be all different factorizations of $\F$ and $\mathcal{K}_{i_j} = \rho(\F_{i_j})$, where $\G_{i_j}\in k[\z]^{l\times r}$, $\F_{i_j}\in k[\z]^{r\times m}$, $j=1,\ldots, t$ and $0\leq t \leq s$ ($t=0$ implies that $\F$ has no factorizations). For each $\mathcal{K}_{i_j}$, if there does not exist $j'$ such that $\mathcal{K}_{i_j} \subsetneq \mathcal{K}_{i_{j'}}$, then $\F = \G_{i_j}\F_{i_j}$ is a FLP factorization of $\F$. The reason is as follows. Assume that there exist $\G_0\in k[\z]^{r\times r}$ and $\F_0\in k[\z]^{r\times m}$ such that $\F_{i_j} = \G_0\F_0$. If ${\rm det}(\G_0) \in k[\z] \setminus k$, then $\mathcal{K}_{i_j} \subsetneq \rho(\F_0)$. It can be seen that $\F = (\G_{i_j}\G_0)\F_0$ is a factorization of $\F$ and it is different from $\F = \G_{i_j}\F_{i_j}$. This contradicts the fact that there exists no $j'$ such that $\mathcal{K}_{i_j} \subsetneq \mathcal{K}_{i_{j'}}$. Then, ${\rm det}(\G_0)$ is a nonzero constant and $\F_{i_j}$ is a FLP matrix.

\vspace{4pt}
 According to the above analysis, we now give a partial solution to Problem \ref{main-problem-2}. We construct the following algorithm to compute all FLP factorizations for $\mathbf{F}\in \mathcal{M}$.

\vspace{6pt}
 Before proceeding further, let us remark on Algorithm \ref{FLP_Algorithm}.

 \begin{enumerate}
   \item[(1)] In step 14 and step 26, we need to compute free bases of free submodules in $k[\z]^{1\times m}$. \cite{Fabianska2007Applications} first designed a Maple package, which is called QUILLENSUSLIN, to implement the Quillen-Suslin theorem. At the same time, they implemented an algorithm for computing free bases of free submodules in this package. Based on this fact, Algorithm \ref{FLP_Algorithm} is implemented on Maple. For interested readers, more examples can be generated by the codes at: \url{http://www.mmrc.iss.ac.cn/~dwang/software.html}.

   \item[(2)] In step 8 and step 20, we need to compute a system of generators of $\mathcal{K}:J$, where $\mathcal{K}\subset k[\z]^{1\times m}$ and $J$ is a nonzero ideal. \cite{Mingsheng2005On} proposed an algorithm to compute $\mathcal{K}:J$, and we have implemented this algorithm on Maple.

   \item[(3)] In step 9 and step 21, if $\F_i'$ is a full row rank matrix, then $\rho(\F_i')$ is a free module of rank $r$ and we do not need to compute a reduced \gr basis of all the $r\times r$ column reduced minors of $\F_i'$; otherwise, we need to use Proposition \ref{free-module-check} to determine whether $\mathcal{K}:J$ is a free module of rank $r$.

   \item[(4)] In step 20, $\rho(\F):(f_i\mathcal{G}) = \rho(\F):\langle f_ic_1,\ldots,f_ic_\xi \rangle$ since $\mathcal{G}$ is a reduced \gr basis of $\langle c_1,\ldots,c_\xi\rangle$. This can help us reduce some calculations.

   \item[(5)] In step 15 and step 27, we need to compute $\G_i\in k[\z]^{l\times r}$ such that $\F = \G_i\F_i$. \cite{Lu2020On} designed a Maple package, which is called poly-matrix-equation, for solving multivariate polynomial matrix Diophantine equations. We use this package to compute $\G_i$.

   \item[(6)] In step 15, Theorem \ref{LWX-main-flp} can guarantee that $d_r(\G_i) = f_i$. In step 27, we can not ensure that $d_r(\G_i) = f_i$. Proposition \ref{LWX-Guan-equivalent} only tell us that there is no a proper divisor $f_i'$ of $f_i$ such that $d_r(\G_i) = f_i'$. Hence, we need to compute $d_r(\G_i)$.

   \item[(7)] In step 25 and step 29, we can use \gr bases to verify the inclusion relationship of two submodules of $k[\z]^{1\times m}$.

   \item[(8)] In step 17, the element $(\F_i',f_i)$ is also deleted since $f_i$ divides itself. Similarly, the element $(\F_i',f_i)$ in step 29 is also deleted since $\rho(\F_i') \subseteq \rho(\F_i')$.

   \item[(9)] In fact, we can obtain all factorizations of $\F$ by making appropriate modifications to Algorithm \ref{FLP_Algorithm}.
 \end{enumerate}
 
  \vskip 12 pt

\begin{algorithm}[H]
 \DontPrintSemicolon
 \SetAlgoSkip{}
 \LinesNumbered
 \SetKwInOut{Input}{Input}
 \SetKwInOut{Output}{Output}

 \Input{$\mathbf{F}\in \mathcal{M}$, the rank $r$ of $\F$ and $d_r(\F)$.}

 \Output{all FLP factorizations of $\mathbf{F}$.}

 \Begin{

  $P: = \emptyset$ and $W: = \emptyset$;

  compute all different divisors $f_1,\ldots,f_s$ of $d_r(\F)$;

  compute all the $r\times r$ column reduced minors $c_1,\ldots,c_\xi$ of $\F$;

  compute a reduced \gr basis $\mathcal{G}$ of $\langle c_1,\ldots,c_\xi\rangle$;

  \If{$\mathcal{G} = \{1\}$}
  {
    \For{$i$ from $1$ to $s$}
    {
      compute a system of generators of $\rho(\F):f_i$, and use all the elements in the system to constitute a matrix $\F_i'\in k[\z]^{s_i\times m}$;

      \If{the reduced \gr basis of all the $r\times r$ column reduced minors of $\F_i'$ is $\{1\}$}
      {
        $P := P \cup \{(\F_i',f_i)\}$;
      }
    }
    \While{$P \neq \emptyset$}
    {
      select any element $(\F_i',f_i)$ from $P$;

      \If{there is no other elements $(\F_j',f_j)\in P$ such that $f_i \mid f_j$}
      {
        compute a free basis of $\rho(\F_i')$, and use all the elements in the basis to constitute a matrix $\F_i\in k[\z]^{r\times m}$;

        compute a matrix $\G_i\in k[\z]^{l\times r}$ such that $\F = \G_i\F_i$;

        $W: = W \cup \{(\G_i,\F_i,f_i)\}$;
      }

      delete all elements $(\F_t',f_t)$ that satisfy $f_t\mid f_i$ from $P$;
    }
  }
  \Else
  {
    \For{$i$ from $1$ to $s$}
    {
      compute a system of generators of $\rho(\F):(f_i\mathcal{G})$, and use all the elements in the system to constitute a matrix $\F_i'\in k[\z]^{s_i\times m}$;

      \If{the reduced \gr basis of all the $r\times r$ column reduced minors of $\F_i'$ is $\{1\}$}
      {
        $P := P \cup \{(\F_i',f_i)\}$;
      }
    }

    \While{$P \neq \emptyset$}
    {
      select any element $(\F_i',f_i)$ from $P$;

      \If{there is no other elements $(\F_j',f_j)\in P$ such that $\rho(\F_i') \subsetneq \rho(\F_j')$}
      {
        compute a free basis of $\rho(\F_i')$, and use all the elements in the basis to constitute $\F_i\in k[\z]^{r\times m}$;

        compute a matrix $\G_i\in k[\z]^{l\times r}$ such that $\F = \G_i\F_i$ with $d_r(\G_i) = f_i'$;

        $W: = W \cup \{(\G_i,\F_i,f_i')\}$;
      }

      delete all elements $(\F_t',f_t)$ that satisfy $\rho(\F_t') \subseteq \rho(\F_i')$ from $P$;
    }
  }

  {\bf return} $W$.
 }
 \caption{FLP factorization algorithm}
 \label{FLP_Algorithm}
 \end{algorithm}

\subsection{Examples}

 We first use the example in \citep{Guan2018} to illustrate the calculation process of Algorithm \ref{FLP_Algorithm}.

 \begin{example}\label{example-1}
  {\rm Let
  \[\mathbf{F} =
  \begin{bmatrix}
       z_1z_2-z_2   & 0  & z_3+1    \\
       0  & z_1z_2-z_2  & z_1^2-2z_1+1 \\
       z_1^2z_2-z_1z_2    &  z_1z_2^2-z_2^2  & z_1^2z_2-2z_1z_2+z_1z_3+z_1+z_2
   \end{bmatrix}\]
  be a multivariate polynomial matrix in $\mathbb{C}[z_1,z_2,z_3]^{3\times 3}$, where $z_1>z_2>z_3$ and $\mathbb{C}$ is the complex field.

  It is easy to compute that the rank of $\F$ is $2$, and $d_2(\F)=(z_1-1)z_2$. Since $d_2(\F)$ is a square-free polynomial, $\F\in \mathcal{M}$. Then, we can use Algorithm \ref{FLP_Algorithm} to compute all FLP factorizations of $\F$. The input of Algorithm \ref{FLP_Algorithm} are $\mathbf{F}$, $r=2$ and $d_2(\F)=(z_1-1)z_2$.

  Let $P = \emptyset$ and $W = \emptyset$. All different divisors of $d_2(\F)$ are: $f_1 = 1$, $f_2 = z_1-1$, $f_3 = z_2$ and $f_4 = (z_1-1)z_2$. All the $2\times 2$ column reduced minors of $\F$ are: $c_1 = 1$, $c_2 = z_2$ and $c_3 = -z_1$. The reduced \gr basis of $\langle c_1,c_2,c_3 \rangle$ w.r.t. the degree reverse lexicographic order is $\mathcal{G} = \{1\}$. Now, we use the steps from 7 to 17 to compute all FLP factorizations of $\F$.

  (1) When $i=1$, we first compute a system of generators of $\rho(\F):f_1$ and the system is $\{[z_1z_2-z_2, ~ 0, ~ z_3+1],~[0, ~ z_1z_2-z_2, ~ z_1^2-2z_1+1]\}$. Let
  \[\mathbf{F}_1' =
  \begin{bmatrix}
   z_1z_2-z_2   &      0     &   z_3+1       \\
   0 & z_1z_2-z_2 & z_1^2-2z_1+1
   \end{bmatrix}.\]
  Since $\rho(\F_1') = \rho(\F):f_1$ and $\F_1'$ is a full row rank matrix, $\rho(\F):f_1$ is a free module of rank $2$.

  (2) When $i=2$, a system of generators of $\rho(\F):f_2$ is $\{[0, ~ z_2, ~ z_1-1],~[z_1z_2-z_2, ~ 0, ~ z_3+1]\}$. Let
  \[\mathbf{F}_2' =
  \begin{bmatrix}
   0   &      z_2     &   z_1-1       \\
   z_1z_2-z_2 & 0 & z_3+1
   \end{bmatrix}.\]
  Since $\rho(\F_2') = \rho(\F):f_2$ and $\F_2'$ is a full row rank matrix, $\rho(\F):f_2$ is a free module of rank $2$.

  (3) When $i=3$, a system of generators of $\rho(\F):f_3$ is $\{[z_1z_2-z_2, ~ 0, ~ z_3+1],~[0, ~ z_1z_2-z_2, ~ z_1^2-2z_1+1], ~[z_1^3-3z_1^2+3z_1-1, ~ -z_1z_3-z_1+z_3+1, ~ 0]\}$. Let
  \[\mathbf{F}_3' =
  \begin{bmatrix}
   z_1z_2-z_2   &      0     &   z_3+1       \\
   0 & z_1z_2-z_2 & z_1^2-2z_1+1  \\
   z_1^3-3z_1^2+3z_1-1  &   -z_1z_3-z_1+z_3+1  &  0
   \end{bmatrix}.\]
  All the $2\times 2$ column reduced minors of $\F_3'$ are $(z_1-1)^2, -z_2$, $z_3+1$. Since $\langle (z_1-1)^2, -z_2,z_3+1 \rangle \neq \mathbb{C}[z_1,z_2,z_3]$, $\rho(\F):f_3$ is not a free module of rank $2$.

  (4) When $i=4$, a system of generators of $\rho(\F):f_4$ is $\{[0, ~ z_2, ~ z_1-1],~[z_1z_2-z_2, ~ 0, ~ z_3+1], ~[z_1^2-2z_1+1, ~ -z_3-1, ~ 0]\}$. Let
  \[\mathbf{F}_4' =
  \begin{bmatrix}
   0   &      z_2     &   z_1-1       \\
   z_1z_2-z_2   &      0     &   z_3+1       \\
   z_1^2-2z_1+1 & -z_3-1  &  0
   \end{bmatrix}.\]
  All the $2\times 2$ column reduced minors of $\F_4'$ are $z_1-1, z_2,z_3+1$. Since $\langle z_1-1, z_2,z_3+1 \rangle \neq \mathbb{C}[z_1,z_2,z_3]$, $\rho(\F):f_4$ is not a free module of rank $2$.

  Then, $P=\{(\F_1',f_1),(\F_2',f_2)\}$. Since $f_2$ is a proper multiple of $f_1$, $\F$ has a FLP factorization w.r.t. $f_2$. Obviously, the rows of $\F_2'$ constitute a free basis of $\rho(\F):f_2$. Let $\F_2 = \F_2'$, we compute a polynomial matrix $\G_2\in \mathbb{C}[z_1,z_2,z_3]^{3\times 2}$ such that
  \[\F = \G_2\F_2 =
    \begin{bmatrix}
     0   &  1   \\
     z_1-1   &  0   \\
     z_1z_2- z_2   &  z_1
   \end{bmatrix}
   \begin{bmatrix}
    0           &   z_2   &    z_1-1           \\
    z_1z_2-z_2  &    0    &    z_3 +1
   \end{bmatrix},\]
   where $d_2(\G_2) = f_2$ and $\F_2$ is a FLP matrix. Then, $W = \{(\G_2,\F_2,f_2)\}$.}
 \end{example}

 \begin{remark}
  Since $\langle c_1,c_2,c_3 \rangle = \langle 1 \rangle$, we can use Theorem \ref{LWX-main-flp-2} to compute all FLP factorizations of $\F$. The above calculation process is simpler than that of Example 3.20 in \citep{Guan2018}. Obviously, Algorithm \ref{FLP_Algorithm} is more efficient than the algorithm proposed in \citep{Guan2018}.
 \end{remark}

 \begin{example}\label{example-2}
  {\rm Let
  \[\mathbf{F} =
  \begin{bmatrix}
    z_1z_2^2  &  z_1z_3^2  &  z_2^2z_3+z_3^3   \\
     z_1z_2   &      0     &      z_2z_3       \\
        0     &  z_1^2z_3  &     z_1z_3^2
   \end{bmatrix}\]
  be a multivariate polynomial matrix in $\mathbb{C}[z_1,z_2,z_3]^{3\times 3}$, where $z_1>z_2>z_3$ and $\mathbb{C}$ is the complex field.

  It is easy to compute that the rank of $\F$ is $2$, and $d_2(\F)=z_1z_2z_3$. Since $d_2(\F)$ is a square-free polynomial, $\F\in \mathcal{M}$. Then, we can use Algorithm \ref{FLP_Algorithm} to compute all FLP factorizations of $\F$. The input of Algorithm \ref{FLP_Algorithm} are $\mathbf{F}$, $r=2$ and $d_2(\F)=z_1z_2z_3$.

  Let $P = \emptyset$ and $W = \emptyset$. All different divisors of $d_2(\F)$ are: $f_1 = 1$, $f_2 = z_1$, $f_3 = z_2$, $f_4 = z_3$, $f_5 = z_1z_2$, $f_6 = z_1z_3$, $f_7 = z_2z_3$ and $f_8 = z_1z_2z_3$. All the $2\times 2$ column reduced minors of $\F$ are: $c_1 = z_1$, $c_2 = z_3$ and $c_3 = z_1z_2$. The reduced \gr basis of $\langle c_1,c_2,c_3 \rangle$ w.r.t. the degree reverse lexicographic order is $\mathcal{G} = \{z_1,z_3\}$. Now, we use the steps from 19 to 29 to compute all FLP factorizations of $\F$.

  Let $\mathcal{K}_i = \rho(\F): \langle f_ic_1,f_ic_2,f_ic_3 \rangle$, where $i=1,\ldots,8$. Since $\mathcal{G}$ is a \gr basis of $\langle c_1,c_2,c_3 \rangle$, for each $i$ we have $\mathcal{K}_i = \rho(\F): \langle f_ic_1,f_ic_2 \rangle = (\rho(\F):f_ic_1) \cap (\rho(\F):f_ic_2)$.

  (1) When $i=1$, the systems of generators of $\rho(\F):z_1$ and $\rho(\F):z_3$ are $\{[z_1z_2, ~ 0, ~ z_2z_3],~[0, ~ z_1z_3,$ $z_3^2],~[-z_2z_3^2, ~ z_2z_3^2, ~ 0]\}$ and $\{[z_1z_2, ~ 0, ~ z_2z_3],~[0, ~ z_1z_3, ~ z_3^2]$, $[0, ~ z_1^2, ~ z_1z_3]\}$, respectively. Then, a system of generators of $\mathcal{K}_1$ is
  \[\{[z_1z_2, ~ 0, ~ z_2z_3],~[0, ~ z_1z_3, ~ z_3^2],~[-z_1z_2z_3^2, ~ z_1z_2z_3^2, ~ 0]\}.\]
  Let
  \[\mathbf{F}_1' =
  \begin{bmatrix}
   z_1z_2   &      0     &      z_2z_3       \\
   0  &  z_1z_3  &  z_3^2   \\
   -z_1z_2z_3^2     &  z_1z_2z_3^2  &     0
   \end{bmatrix}.\]
  It is easy to compute that all the $2\times 2$ column reduced minors of $\F_1'$ are $1, z_2z_3,z_3^2$. Since $\langle 1, z_2z_3,z_3^2 \rangle = \mathbb{C}[z_1,z_2,z_3]$, $\mathcal{K}_1$ is a free module of rank $2$.

  (2) When $i=2$, the systems of generators of $\rho(\F):z_1^2$ and $\rho(\F):z_1z_3$ are $\{[z_1z_2, ~ 0, ~ z_2z_3],$ $[0, ~ z_1z_3, ~ z_3^2],~[z_2z_3, ~ -z_2z_3, ~ 0]\}$ and $\{[z_1z_2, ~ 0, ~ z_2z_3],~[0, ~ z_1, ~ z_3]$, $[z_2z_3, ~ -z_2z_3, ~ 0]\}$, respectively. Then, a system of generators of $\mathcal{K}_2$ is
  \[\{[z_1z_2, ~ 0, ~ z_2z_3],~[0, ~ z_1z_3, ~ z_3^2],~[z_2z_3, ~ -z_2z_3, ~ 0]\}.\]
  Let
  \[\mathbf{F}_2' =
  \begin{bmatrix}
   z_1z_2   &      0     &      z_2z_3       \\
   0  &  z_1z_3  &  z_3^2   \\
   z_2z_3     &  -z_2z_3  &     0
   \end{bmatrix}.\]
  It is easy to compute that all the $2\times 2$ column reduced minors of $\F_2'$ are $z_1,-z_2,-z_3$. Since $\langle z_1,-z_2,-z_3 \rangle \neq \mathbb{C}[z_1,z_2,z_3]$, $\mathcal{K}_2$ is not a free module of rank $2$.

 (3) According to the above same steps, we have that the systems of generators of $\mathcal{K}_3, \ldots, \mathcal{K}_8$ are $\{[z_1, ~ 0, ~ z_3],~[0, ~ z_1z_3, ~ z_3^2]\}$, $\{[0, ~ z_1, ~ z_3],~[z_1z_2, ~ 0, ~ z_2z_3]\}$, $\{[-z_3, ~ z_3, ~ 0],~[z_1, ~ 0, ~ z_3]\}$, $\{[0, ~ z_1, ~ z_3],$ $[z_2, ~ -z_2, ~ 0]\}$, $\{[z_1, ~ 0, ~ z_3],~[0, ~ z_1, ~ z_3]\}$ and $\{[0, ~ z_1, ~ z_3],~[-1, ~ 1, ~ 0]\}$, respectively. Let $\F_i'\in \mathbb{C}[z_1,z_2,z_3]^{2\times 3}$ be composed of the above system of generators of $\mathcal{K}_i$, where $i =3,\ldots,8$. For each $i$, it is easy to compute that ${\rm rank}(\F_i')=2$. This implies that $\F_i'$ is a full row rank matrix. Then, $\mathcal{K}_i = \rho(\F_i')$ is a free module of rank $2$. Then, we have
 \[P = \{(\F_1',f_1),(\F_3',f_3),\ldots,(\F_8',f_8)\}.\]

 (4) Since $\rho(\F_i') \subsetneq \rho(\F_8')$ for each $1\leq i \leq 7$ with $i \neq 2$, $\F$ has only one FLP factorization. Since
  \[\mathbf{F}_8' =
  \begin{bmatrix}
   0  &      z_1     &   z_3       \\
   -1  &      1      &    0
   \end{bmatrix}\]
  is a full row rank matrix, the rows of $\mathbf{F}_8'$ constitute a free basis of $\mathcal{K}_8 = \rho(\F_8')$. Let $\mathbf{F}_8 = \mathbf{F}_8'$, we compute a polynomial matrix $\G_8\in \mathbb{C}[z_1,z_2,z_3]^{3\times 2}$ such that
  \[\F = \G_8\F_8 =
    \begin{bmatrix}
     z_2^2+z_3^2   &  -z_1z_2^2   \\
      z_2   &  -z_1z_2   \\
      z_1z_3    &  0
   \end{bmatrix}
   \begin{bmatrix}
   0  &      z_1     &   z_3       \\
   -1  &      1      &    0
   \end{bmatrix},\]
   where $\F_8$ is a FLP matrix. It is easy to compute that $d_2(\G_8) = f_8$. Then, $W = \{(\G_8,\F_8,f_8)\}$.}
 \end{example}

\section{Concluding Remarks}\label{sec_conclusions}

 In this paper we have studied two FLP factorization problems for multivariate polynomial matrices without full row rank. As we all know, FLP factorizations are still open problems so far. In order to solve some special situations, we have introduced the concept of column reduced minors. Then, we have proved a theorem which provides a necessary and sufficient condition for a class of multivariate polynomial matrices without full row rank to have FLP factorizations. Moreover, we have given a simple method to verify whether a submodule of $k[\z]^{1\times m}$ is a free module by using column reduced minors of polynomial matrices. Compared with the traditional method, the new method is more efficient. Based on our results, we have also proposed an algorithm for FLP factorizations and have implemented it on the computer algebra system Maple. Two examples have been given to illustrate the  effectiveness of the algorithm.

 Let $\F\in k[\z]^{l\times m}$, every full column rank submatrix of $\F$ is a square matrix if ${\rm rank}(\F) = l$. In this case, all the $l\times l$ column reduced minors of $\F$ are only one polynomial which  is the constant $1$. Therefore, all the results in this paper are also valid for the case where $\F$ is a full row rank matrix.

 We can define the concept of row reduced minors, and all the results in this paper can be translated to similar results for FRP factorizations of multivariate polynomial matrices without full column rank. We hope the results provided in the paper will motivate further research in the area of factor prime factorizations.

\section*{Acknowledgments}

 This research was supported by the CAS Key Project QYZDJ-SSW-SYS022.

\bibliographystyle{elsarticle-harv}

\bibliography{FLP_Manuscript}

\begin{thebibliography}{24}
\expandafter\ifx\csname natexlab\endcsname\relax\def\natexlab#1{#1}\fi
\expandafter\ifx\csname url\endcsname\relax
  \def\url#1{\texttt{#1}}\fi
\expandafter\ifx\csname urlprefix\endcsname\relax\def\urlprefix{URL }\fi

\bibitem[{Bose(1982)}]{Bose1982}
Bose, N., 1982. {\em Applied {M}ultidimensional {S}ystems {T}heory}. Van
  Nostrand Reinhold, New York.

\bibitem[{Bose et~al.(2003)Bose, Buchberger, and Guiver}]{Bose2003}
Bose, N., Buchberger, B., Guiver, J., 2003. {\em Multidimensional {S}ystems
  {T}heory and {A}pplications}. Dordrecht, The Netherlands: Kluwer.

\bibitem[{Cox et~al.(2005)Cox, Little, and O'shea}]{Cox2005Using}
Cox, D., Little, J., O'shea, D., 2005. {\em Using {A}lgebraic {G}eometry}.
  Graduate Texts in Mathematics (Second Edition). Springer, New York.

\bibitem[{Cox et~al.(2007)Cox, Little, and O'shea}]{Cox2007Ideals}
Cox, D., Little, J., O'shea, D., 2007. {\em Ideals, {V}arieties, and
  {A}lgorithms, third edition}. Graduate Texts in Mathematics (Second Edition).
  Springer, New York.

\bibitem[{Eisenbud(2013)}]{Eisenbud2013}
Eisenbud, D., 2013. {\em Commutative {A}lgebra: with a view toward algebraic
  geometry}. New York: Springer.

\bibitem[{Fabia\'{n}ska and Quadrat(2007)}]{Fabianska2007Applications}
Fabia\'{n}ska, A., Quadrat, A., 2007. Applications of the {Q}uillen-{S}uslin
  theorem to multidimensional systems theory. In: Park, H., Regensburger, G.,
  (Eds.), {\em \gr Bases in Control Theory and Signal Processing}, Radon Series
  on Computational and Applied Mathematics 3, 23--106.

\bibitem[{Greuel and Pfister(2002)}]{Greuel2002A}
Greuel, G., Pfister, G., 2002. {\em A {SINGULAR} {I}ntroduction to
  {C}ommutative {A}lgebra}. Springer-Verlag.

\bibitem[{Guan et~al.(2018)Guan, Li, and Ouyang}]{Guan2018}
Guan, J., Li, W., Ouyang, B., 2018. On rank factorizations and factor prime
  factorizations for multivariate polynomial matrices. {\em Journal of Systems
  Science and Complexity} 31~(6), 1647--1658.

\bibitem[{Guan et~al.(2019)Guan, Li, and Ouyang}]{Guan2019}
Guan, J., Li, W., Ouyang, B., 2019. On minor prime factorizations for
  multivariate polynomial matrices. {\em Multidimensional Systems and Signal
  Processing} 30, 493--502.

\bibitem[{Guiver and Bose(1982)}]{Guiver1982Polynomial}
Guiver, J., Bose, N., 1982. Polynomial matrix primitive factorization over
  arbitrary coefficient field and related results. {\em IEEE Transactions on
  Circuits and Systems} 29~(10), 649--657.

\bibitem[{Lin(1988)}]{Lin1988On}
Lin, Z., 1988. On matrix fraction descriptions of multivariable linear $n$-{D}
  systems. {\em IEEE Transactions on Circuits and Systems} 35~(10), 1317--1322.

\bibitem[{Lin(1999)}]{Lin1999Notes}
Lin, Z., 1999. Notes on $n$-{D} polynomial matrix factorizations. {\em
  Multidimensional Systems and Signal Processing} 10~(4), 379--393.

\bibitem[{Lin and Bose(2001)}]{Lin2001A}
Lin, Z., Bose, N., 2001. A generalization of {S}erre's conjecture and some
  related issues. {\em Linear Algebra and Its Applications} 338~(1), 125--138.

\bibitem[{Liu and Wang(2015)}]{Liu2015Further}
Liu, J., Wang, M., 2015. Further remarks on multivariate polynomial matrix
  factorizations. {\em Linear Algebra and Its Applications} 465~(465),
  204--213.

\bibitem[{Lu et~al.(2020)Lu, Wang, and Xiao}]{Lu2020On}
Lu, D., Wang, D., Xiao, F., 2020. {\em poly-matrix-equation: a {M}aple package,
  for solving multivariate polynomial matrix {D}iophantine equations}.
  \url{http://www.mmrc.iss.ac.cn/~dwang/software.html}.

\bibitem[{Morf et~al.(1977)Morf, Levy, and Kung}]{Morf1977New}
Morf, M., Levy, B., Kung, S., 1977. New results in 2-{D} systems theory, part
  {I}: 2-{D} polynomial matrices, factorization, and coprimeness. {\em
  Proceedings of the IEEE} 64~(6), 861--872.

\bibitem[{Pommaret(2001)}]{Pommaret2001Solving}
Pommaret, J., 2001. Solving {B}ose conjecture on linear multidimensional
  systems. In: {\em European Control Conference}. IEEE, Porto, Portugal, pp.
  1653--1655.

\bibitem[{Quillen(1976)}]{Quillen1976Projective}
Quillen, D., 1976. Projective modules over polynomial rings. {\em Inventiones
  Mathematicae} 36, 167--171.

\bibitem[{Sule(1994)}]{Sule1994Feed}
Sule, V., 1994. Feedback stabilization over commutative rings: the matrix case.
  {\em SIAM Journal on Control and Optimization} 32~(6), 1675--1695.

\bibitem[{Suslin(1976)}]{Suslin1976Projective}
Suslin, A., 1976. Projective modules over polynomial rings are free. {\em
  Soviet Mathematics - Doklady} 17, 1160--1165.

\bibitem[{Wang(2007)}]{Mingsheng2007On}
Wang, M., 2007. On factor prime factorization for $n$-{D} polynomial matrices.
  {\em IEEE Transactions on Circuits and Systems--I: Regular Papers} 54~(6),
  1398--1405.

\bibitem[{Wang and Feng(2004)}]{Wang2004On}
Wang, M., Feng, D., 2004. On {L}in-{B}ose problem. {\em Linear Algebra and Its
  Applications} 390~(1), 279--285.

\bibitem[{Wang and Kwong(2005)}]{Mingsheng2005On}
Wang, M., Kwong, C., 2005. On multivariate polynomial matrix factorization
  problems. {\em Mathematics of Control, Signals, and Systems} 17~(4),
  297--311.

\bibitem[{Youla and Gnavi(1979)}]{Youla1979Notes}
Youla, D., Gnavi, G., 1979. Notes on $n$-dimensional system theory. {\em IEEE
  Transactions on Circuits and Systems} 26~(2), 105--111.

\end{thebibliography}

\end{document}